\documentclass[11pt,a4paper]{article}
\pdfoutput=1

\usepackage[vmargin=33mm,hmargin=33mm]{geometry}
\usepackage{amsmath,amssymb,amsthm,booktabs,color,enumitem,graphicx,latexsym,url,xspace}
\usepackage[numbers,sort&compress]{natbib}
\usepackage[basic,small]{complexity}
\usepackage[small]{caption}
\usepackage{microtype}
\usepackage{hyperref}

\newtheorem{lemma}{Lemma}

\hypersetup{
    colorlinks=true,
    linkcolor=black,
    citecolor=black,
    filecolor=black,
    urlcolor=[rgb]{0,0.1,0.5},
    pdfauthor={Juhana Laurinharju, Jukka Suomela},
    pdftitle={Linial's lower bound made easy},
}

\begin{document}

\begin{flushleft}
{\Large\bfseries Linial's Lower Bound Made Easy\par}

\bigskip
\textbf{Juhana Laurinharju}
$\cdot$
{\sffamily\small juhana.laurinharju@cs.helsinki.fi}

{\footnotesize
Helsinki Institute for Information Technology HIIT, \\
Department of Computer Science, University of Helsinki, Finland\par}

\medskip
\textbf{Jukka Suomela}
$\cdot$
{\sffamily\small jukka.suomela@aalto.fi}

{\footnotesize
Helsinki Institute for Information Technology HIIT, \\
Department of Information and Computer Science, Aalto University, Finland\par}
\end{flushleft}

\bigskip
\bigskip
\noindent\textbf{Abstract.}
Linial's seminal result shows that any deterministic distributed algorithm that finds a $3$-colouring of an $n$-cycle requires at least $\log^*(n)/2 - 1$ communication rounds. We give a new simpler proof of this theorem.

\medskip

\section{Introduction}

Linial's lower bound for $3$-colouring directed cycles~\cite{linial92locality} is one of the most celebrated results in the area of distributed graph algorithms. It is cited in hundreds of papers and the proof has been reproduced in textbooks and lecture notes \cite{peleg00distributed,barenboim13distributed,wattenhofer13lecture,suomela-ddabook}. Yet it it seems that typical presentations of this result either follow the structure of Linial's original proof~\cite{peleg00distributed,barenboim13distributed,wattenhofer13lecture}, or rely on some prior knowledge of Ramsey's theorem~\cite{suomela-ddabook}.

In this work we give a simpler, self-contained version of Linial's proof. This version of the proof is easy to explain to a student on a whiteboard in fifteen minutes. We do not need to refer to neighbourhood graphs, line graphs, and chromatic numbers.

\section{Problem Formulation}

Fix a natural number $n$. We are interested in deterministic distributed algorithms that find a proper $3$-colouring of any directed $n$-cycle. The nodes are labelled with unique identifiers from the set $\{1,2,\dotsc,n\}$. Each node must pick its own colour from the set $\{1,2,3\}$.

If we have a distributed algorithm with a running time of $T$ communication rounds, then each node has to pick its own colour based on the information that is available within distance $T$ from it; see Figure~\ref{fig:A}. Moreover, two nodes that are adjacent to each other must pick different colours. Hence the algorithm is a function $A$ with $2T+1$ arguments that satisfies
\begin{align*}
    A(x_1, x_2, \dotsc, x_{2T+1}) &\in \{1,2,3\}, \\
    A(x_1, x_2, \dotsc, x_{2T+1}) &\ne A(x_2, x_3, \dotsc, x_{2T+2})
\end{align*}
whenever $x_1, x_2, \dotsc, x_{2T+2}$ are distinct identifiers from the set $\{1,2,\dotsc,n\}$.

Function $\log^* x$ is the iterated logarithm of $x$, defined as follows:
$\log^* x = 0$ if $x \le 1$, and
$\log^* x = 1 + \log^* \log_2 x$ otherwise.
Linial's famous result shows that no matter which algorithm $A$ we pick, we must have
\begin{equation}
    T \ge \frac{1}{2}\log^*(n) - 1. \label{eq:main}
\end{equation}
We will now give a simple proof of this theorem.

\section{Colouring Functions}

The only concept that we need is a \emph{colouring function}. We say that $A$ is a $k$-ary $c$-colouring function if
\begin{align}
    A(x_1, x_2, \dotsc, x_k) &\in \{1,2,\dotsc,c\} & \text{for all } & 1 \le x_1 < x_2 < \dotso < x_k \le n, \label{eq:col1} \\
    A(x_1, x_2, \dotsc, x_k) &\ne A(x_2, x_3, \dotsc, x_{k+1}) & \text{for all } & 1 \le x_1 < x_2 < \dotso < x_{k+1} \le n. \label{eq:col2}
\end{align}
Any deterministic distributed algorithm $A$ that finds a proper $3$-colouring of an $n$-cycle defines a $k$-ary $3$-colouring function for $k = 2T+1$ (the converse is not necessarily true).

We will show that $k + 1 \ge \log^* n$ for any $k$-ary $3$-colouring function. By plugging in $k = 2T+1$, we obtain the main result~\eqref{eq:main}.

\section{Proof}

The proof is by induction; the base case is trivial. If a colouring function only sees $1$ identifier, it cannot do much.

\begin{lemma}\label{lem:base}
    If $A$ is a $1$-ary $c$-colouring function, we have $c \ge n$.
\end{lemma}
\begin{proof}
    If $c < n$, by the pigeonhole principle there are some $x_1 < x_2$ with $A(x_1) = A(x_2)$, which contradicts \eqref{eq:col2}.
\end{proof}

The key part of the proof is the inductive step. Given any colouring function~$A$, we can always construct another colouring function~$B$ that is ``faster'' (smaller number of arguments) but ``worse'' (larger number of colours). Here it is crucial that colouring functions are well-defined for both odd and even values of~$k$.

\begin{lemma}\label{lem:ind}
    If $A$ is a $k$-ary $c$-colouring function, we can construct a $(k-1)$-ary $2^c$-colouring function~$B$.
\end{lemma}
\begin{proof}
    We define $B$ as follows:
    \[
        B(x_1, x_2, \dotsc, x_{k-1}) = \bigl\{ A(x_1, x_2, \dotsc, x_{k-1}, x_k) \,:\, x_k > x_{k-1} \bigr\}.
    \]
    There are only $2^c$ possible values of $B$: all possible subsets of $\{1,2,\dotsc,c\}$. These can be represented as integers $\{1,2,\dotsc,2^c\}$, and hence \eqref{eq:col1} holds.

    The interesting part is \eqref{eq:col2}. Let $1 \le x_1 < x_2 < \dotso < x_k \le n$. By way of contradiction, suppose that
    \begin{equation}
        B(x_1, x_2, \dotsc, x_{k-1}) = B(x_2, x_3, \dotsc, x_k). \label{eq:contr}
    \end{equation}
    Let
    \[
        \alpha = A(x_1, x_2, \dotsc, x_k).
    \]
    From the definition of $B$ we have
    $\alpha \in B(x_1, x_2, \dotsc, x_{k-1})$.
    By assumption \eqref{eq:contr}, this implies
    $\alpha \in B(x_2, x_3, \dotsc, x_k)$.
    But then we must have some $x_k < x_{k+1} \le n$ such that
    \[
        \alpha = A(x_2, x_3, \dotsc, x_{k+1}).
    \]
    That is, $A$ cannot be a colouring function.
\end{proof}

To complete the proof, we will need power towers. Define
\[
    {}^i 2 = 2^{2^{\cdot^{\cdot^2}}}
\]
with $i$ twos in the power tower. For example, ${}^2 2 = 4$ and ${}^3 2 = 16$. Now assume that $A_1$ is a $k$-ary $3$-colouring function. Certainly it is also a $k$-ary ${}^2 2$-colouring function. We can apply Lemma~\ref{lem:ind} iteratively to obtain
\begin{itemize}[noitemsep]
    \item a $(k-1)$-ary ${}^3 2$-colouring function $A_2$,
    \item a $(k-2)$-ary ${}^4 2$-colouring function $A_3$, \\ \ldots
    \item a $1$-ary ${}^{k+1} 2$-colouring function $A_k$.
\end{itemize}
By Lemma~\ref{lem:base}, we must have ${}^{k+1} 2 \ge n$, which implies $k + 1 \ge \log^* n$.

\begin{figure}[b]
    \centering
    \includegraphics[page=1]{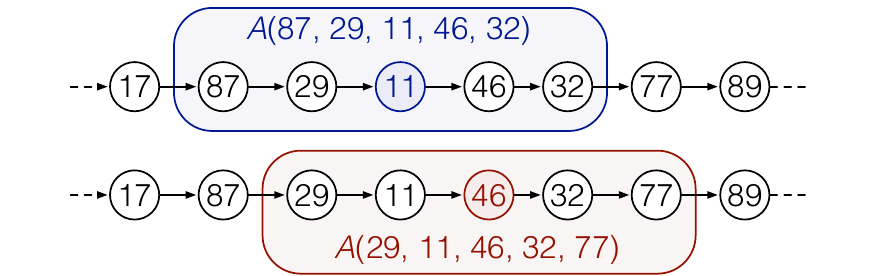}
    \caption{Colouring directed cycles in time $T = 2$. For example, the output of node $11$ only depends on its radius-$T$ neighbourhood, $(87,29,11,46,32)$. We can interpret algorithm $A$ as a $k$-ary function, $k = 2T+1 = 5$, that maps each local neighbourhood to a colour. As it is possible that adjacent nodes have neighbourhoods $(87,29,11,46,32)$ and $(29,11,46,32,77)$, function $A$ must satisfy $A(87,29,11,46,32) \ne A(29,11,46,32,77)$.}\label{fig:A}
\end{figure}

\section*{Acknowledgements}

This work was supported in part by the Academy of Finland, Grant 252018, and by the Research Funds of the University of Helsinki. Many thanks to all students who have kindly served as guinea pigs.

\def\UrlFont{\sf\footnotesize}
\setlength{\bibsep}{3pt}
\bibliographystyle{plainnat}
\bibliography{linial-easy}

\end{document}